\documentclass[a4paper,USenglish,cleveref,autoref]{lipics-v2021}

\usepackage{algorithm}
\usepackage{algpseudocode}

\newcommand{\ie}{{\it i.e.,~}}
\newcommand{\eg}{{\it e.g.,~}}
\newcommand{\etal}{{\it et al.}}

\newcommand{\LL}{\emph{Look}}
\newcommand{\CC}{\emph{Compute}}
\newcommand{\MM}{\emph{Move}}

\newcommand{\FSYNC}{\mathsf{FSYNC}}

\newcommand{\isMulti}{\mathsf{isMulti}}

\newcommand{\opset}{\mathsf{OPSET}}

\newcommand{\TRUE}{\mathsf{TRUE}}
\newcommand{\FALSE}{\mathsf{FALSE}}

\title{Gathering Despite Defected View} 

\titlerunning{Gathering Despite Defected View} 

\author{Yonghwan Kim}
{Nagoya Institute of Technology, Aichi, Japan}
{kim@nitech.ac.jp}{https://orcid.org/0000-0002-5437-7626}{}

\author{Masahiro Shibata}
{Kyushu Institute of Technology, Fukuoka, Japan}
{shibata@csn.kyutech.ac.jp}{https://orcid.org/0000-0003-1414-8033}{}

\author{Yuichi Sudo}
{Hosei University, Tokyo, Japan}
{sudo@hosei.ac.jp}{https://orcid.org/0000-0002-4442-1750}{}

\author{Junya Nakamura}
{Toyohashi University of Technology, Aichi, Japan}
{junya@imc.tut.ac.jp}{https://orcid.org/0000-0002-1363-4358}{}

\author{Yoshiaki Katayama}
{Nagoya Institute of Technology, Aichi, Japan}
{katayama@nitech.ac.jp}{https://orcid.org/0000-0003-1683-2154}{}

\author{Toshimitsu Masuzawa}
{Osaka University, Osaka, Japan}
{masuzawa@ist.osaka-u.ac.jp}{https://orcid.org/0000-0003-4628-6393}{}

\authorrunning{Y. Kim et al.} 

\Copyright{Yonghwan Kim} 



\category{} 

\relatedversion{} 

\supplement{} 

\acknowledgements{This work was supported in part by JSPS KAKENHI Grant Numbers 18K18031, 19H04085, 19K11823, 20H04140, 20KK0232, 21K17706, and Foundation of Public Interest of Tatematsu.}

\nolinenumbers 

\begin{document}

\maketitle

\begin{abstract}
An autonomous mobile robot system consisting of many mobile computational entities (called \emph{robots}) 
attracts much attention of researchers, and
to clarify the relation between the capabilities of robots and solvability of the problems 
is an emerging issue for a recent couple of decades. 

Generally, each robot can observe all other robots 
as long as there are no restrictions for visibility range or obstructions,
regardless of the number of robots.
In this paper, we provide a new perspective on the observation by robots;
a robot cannot necessarily observe all other robots 
regardless of distances to them.
We call this new computational model \emph{defected view model}.
Under this model, in this paper, we consider the \emph{gathering} problem that requires all the robots to gather 
at the same point and propose two algorithms to solve the gathering problem
in the adversarial ($N$,$N-2$)-defected model for $N \geq 5$ 
(where each robot observes at most $N-2$ robots chosen adversarially) 
and the distance-based (4,2)-defected model 
(where each robot observes at most 2 closest robots to itself) respectively, 
where $N$ is the number of robots.
Moreover, we present an impossibility result showing that
there is no (deterministic) gathering algorithm 
in the adversarial or distance-based (3,1)-defected model.
Moreover, we show an impossibility result for the gathering in a relaxed ($N$, $N-2$)-defected model.
\end{abstract}

\section{Introduction}
\label{sec:intro}
An autonomous mobile robot system, which is firstly introduced in \cite{SY99}, 
is a distributed system which consists of many mobile computational entities 
(called \emph{robots}) with limited capabilities,
\eg robots do not have identifiers, cannot distinguish other robots, 
or cannot remember their any past actions.
The robots operate autonomously and cooperatively;
each robot observes the other robots (called \LL), 
computes the destination based on the observation result (called \CC), 
and moves to the destination point (called \MM).
Each robot autonomously and cyclically performs the above three operations, \LL, \CC, and \MM, 
to achieve the given common goal.
The literature \cite{gatherimpos,perspec13,distalg06,SY99} provide a formal discussion on the capabilities of the robots 
for the distributed coordination (\eg gathering, scattering, or pattern formation).
Since the introduction of the autonomous mobile robot system, 
many related studies about the computational power of the system have been introduced 
and many researchers are interested in clarifying the relationship between the capabilities of the robots 
and solvability of the problems.

Generally, in \LL~operation, 
each robot can observe all other robots (within its visibility range if it has a limited visibility range)
to compute the destination point to move.
In other words, each robot takes a snapshot that consists of all other robots' (relative) positions in its \LL~operation.
This implies that each robot temporarily remembers the positions of up to 
$N-1$ robots, 
where $N$ is the total number of robots.
From the more practical viewpoint, 
we claim that a robot with low functionality may not have such large working memory.
That is, the main question we address in this paper is ``\emph{what occurs if a robot cannot observe some of other robots?}''.
More precisely, ``\emph{how many other robots should be observed to achieve the goals of the problems?}''.

\vspace{5pt}
\noindent
{\bf Related works:}
The gathering problem \cite{gatherfea},
which requires all the robots move to a common (non-predetermined) position,
is one of the basic and essential problems in autonomous mobile robot systems.
There are many studies about the gathering of the autonomous mobile robots;
M. Cieliebak \etal~presented the first algorithm to achieve the gathering from any arbitrary configuration \cite{gatheranyconf},
R. Klasing \etal~studied the gathering of mobile robots in one node of an anonymous unoriented ring \cite{gatherring},
G. D'Angelo \etal~introduced a gathering algorithm by robots without multiplicity detection on grids and trees \cite{gathergrid},
and many works for the gathering by robots with dynamic (or inaccurate) compass are also introduced 
\cite{gatherinac3,gatherinac2,gatherinac4,gatherinac1}.
The capability of the robots deeply affects the solvability of the gathering problem, thus 
some investigations about the required capability or impossibility are introduced \cite{gatherfea,gatherimpos}.
However, all of these works assume that each robot can observe all other robots 
within its visibility range if there is no obstruction between the robots \eg any opaque robot.

The works most related to this paper is those with the limited visibility range \cite{gathernulimited,gatherlimited,convlimited,circlelimited}.
The robots with the limited visibility cannot necessarily observe all robots, which is  similar to the defected view model we introduce.  But visibility is limited by distance in the model and thus all robots can be observed when they gather closely enough.
On the other hand, our model cannot guarantee such a full view of the robots.

The works for fault-tolerance \cite{gathermulfault,gatherfault,gatherweakmulfault} are also closely related to this paper.
The defected view in our model can be considered as a new type of faults in autonomous mobile robot systems.

\vspace{5pt}
\noindent
{\bf Contribution:}
To provide some answers for the above research questions, 
we propose a new computational model 
with restriction on the number of the robots that each robot can observe,
named the \emph{defected view model}, 
where each robot observes only $k$ other robots for $1 \leq k < N-1$.
The assumption of this model naturally arises
by considering some issues for robots with low functionality 
such as 
(1) each robot does not have enough working memory to store the entire observation result,
(2) each robot may miss some of observation results due to memory failure, 
or (3) each robot fails to observe some of other robots by sensing failure.
It is obvious that when $k$ becomes the lower, the problem becomes the harder
(possibly impossible) to solve.
We consider two different defected view models 
regarding which $k$ robots are observed:
the distance-based ($N$,$k$)-defected model 
and the adversarial ($N$,$k$)-defected model.
In the former model, each robot observes the $k$ closest robots to its current position,
and in the latter one, $k$ robots which are observed by each robot are determined adversarially.

More precisely, the $k$ robots that each robot can observe are chosen from the robots located at points different from $r$'s current position.
Concerning $r$'s current position, $r$ can detect only whether another robot exists at the point or not (so called the \emph{weak multiplicity detection}).
Such an assumption that the robots at $r$'s current position are excluded from the candidates of the observed $k$ robots is motivated from the following observation: 
each robot $r$ observes the robots at remote points and those at $r$'s current position by different ways usually. 
Each robot observes the remote robots by, for example, a radar sensor or a vision senor, 
but senses the other robots at the same point by, for example, a contact sensor.

As the first step of the gathering in the defected view model, 
we investigate only the case of $k=N-2$. 
The main contributions of this paper is as follows:
(1) we propose a gathering algorithm in the adversarial ($N$,$N-2$)-defected model for any $N \geq 5$,
(2) we present another algorithm to solve the gathering problem in the distance-based (4,2)-defected model,
and (3) we provide the impossibility result showing that there is no (deterministic) algorithm to solve the gathering problem 
in the adversarial or distance-based 
(3,1)-defected model.
Moreover, we present another impossibility result 
in a naturally relaxed ($N$,$N-2$)-defected model where the observed $k$
robots can contain the robots at the observer's current position.
This impossibility result shows the necessity of the assumption that 
the observed $N-2$ robots should be chosen from robots other than
those located at the observer's current position.

The rest of this paper is organized as follows:
Section \ref{sec:model} presents the system model (two defected view models) and problem definition;
Section \ref{sec:nn-alg} introduces an algorithm to solve the gathering problem in the adversarial ($N$, $N-2$)-defected model 
for any $N \geq 5$;
Section \ref{sec:42alg} gives a gathering algorithm in the distance-based (4,2)-defected model;
Section \ref{sec:impos} shows two impossibility results
showing that there is no algorithm in adversarial or
distance-based (3,1)-defected model 
and the relaxed adversarial ($N$,$N-2$)-defected model;
and Section \ref{sec:conclusion} concludes the paper and provides some open problems.

\section{Model and Problem Definition}
\label{sec:model}
\subsection{Robots}
Let $R=\{r_1, r_2, ..., r_N\}$ be the set of $N$ autonomous mobile robots deployed in a plane.
Robots are indistinguishable by their appearance (\ie identical), 
execute the same algorithm (\ie uniform or homogeneous), 
and have no memory (\ie oblivious). 
There is no geometrical agreement; 
robots do not agree on any axis, the unit distance, nor chirality.

A point in the plane is called an \emph{occupied point} if there exists a robot at the point. 
We allow two or more robots to occupy the same point at the same time.
We call a robot a \emph{single robot} if the point occupied by the robot has no other robot.  Otherwise, we call it an \emph{accompanied robot}.

Each robot cyclically performs the three operations, \LL, \CC, and \MM:
{\bf (\LL)} 
a robot obtains the positions (based on its local coordinate system centered on itself) 
of all other (visible) robots, 
{\bf (\CC)} 
a robot determines the destination according to the given algorithm based on the result of \LL~operation. 
Since each robot has no memory, the result of \CC~is determined only by the result of \LL~operation, 
and 
{\bf (\MM)}
a robot moves to the destination computed in \CC~operation. 
We assume \emph{rigid movement} which ensures 
each robot can reach the destination during its \MM~operation, \ie a robot never stops before it reaches its destination.

\subsection{Schedule and Configuration}
We assume a fully-synchronous scheduler ($\FSYNC$): 
all robots fully-synchronously perform their operations (\LL, \CC, and \MM).
This means that all robots perform the same operation at the same time instant and duration.
We call the time duration 
in which all robots perform the three operations (\LL, \CC, and \MM) once 
\emph{a round}.

Let \emph{configuration} $C_t$ be the set of the (global) coordinates of all robots at a given time $t$: 
$C_t = \{(r_{1.x}^t,r_{1.y}^t), (r_{2.x}^t,r_{2.y}^t) \ldots$ $(r_{N.x}^t,r_{N.y}^t)\}$, 
where $r_{i.x}^t$ (resp. $r_{i.y}^t$) is the $X$-coordinate
(resp. $Y$-coordinate) of robot $r_i$ at time $t$.
Note that no robot knows its global coordinate.
Configuration $C_t$ is changed into another configuration $C_{t+1}$ 
after one round (\ie all robots execute the three operations once).

\subsection{Observation: Visibility Range and Multiplicity Detection}
We basically assume that every robot has unlimited visibility range,
\ie any two robots can observe each other regardless of their distance, while we introduce in Definition \ref{def:model} the defect in the information obtained by \LL\ operation.
Moreover, we assume a \emph{weak multiplicity detection},
\ie each robot cannot get the exact number of robots occupying the same point but can distinguish whether the point is occupied by one robot or by multiple robots.
This implies that when each robot observes any point,
it can distinguish the three cases: there is \emph{no robot}, \emph{one robot}, or \emph{two or more robots} at the point.

We consider a \emph{defected view} such that each robot may not observe all other robots.
We define the ($N$,$k$)-defected model, where $1 \leq k < N$ as follows:

\begin{definition}
\label{def:model}
{\bf (($N$,$k$)-defected model)} 
Each robot $r$ can get from \LL\ operation the set of occupied points 
(in its coordinate system) where $k$ robots not accompanied with $r$ are located
(\ie the $k$ robots contains no robot located at $r$'s current point). 
When the number of robots not accompanied with $r$ is less than $k$, 
all such robots are observed. 
The weak multiplicity detection concerning the $k$ robots is assumed: 
a point occupied by only one of the $k$ robots can be distinguished from that occupied by two or more of the $k$ robots.
Moreover, $r$ can distinguish whether $r$ is single or accompanied.
\end{definition}

Note that the ($N$,$N-1$)-defected model is equivalent to the commonly used model (with the weak multiplicity detection) where each robot can observe all robots.
The ($N$,$k$)-defected model has options depending on how the observed $k$ robots are chosen.
We consider the two options in this paper,
named \emph{adversarial ($N$,$k$)-defected model} and \emph{distance-based ($N$,$k$)-defected model}.
In the \emph{adversarial ($N$,$k$)-defected model}, 
$k$ robots observed by each robot are determined adversarially.
In the \emph{distance-based ($N$,$k$)-defected model}, 
each robot $r$ observes the $k$ closest robots to the $r$'s current point.
Tie breaks among the robots the same distance apart is determined in an arbitrary way.

\begin{figure}[b]
	\begin{center}
	    \includegraphics[scale=0.8]{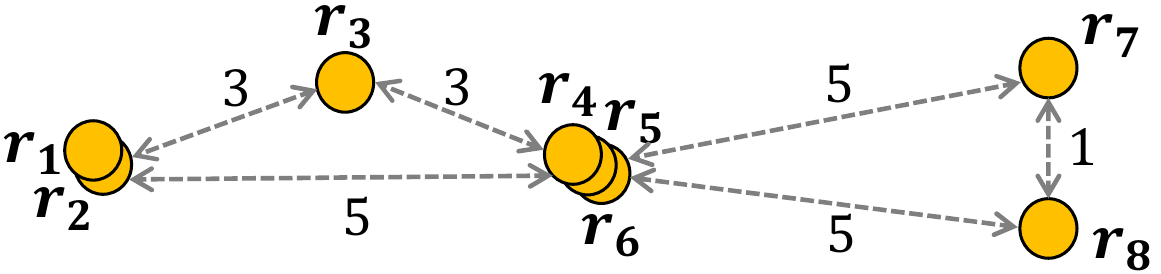}
	\end{center}
    \caption{An example configuration by 8 robots}
    \label{fig:exconf}
\end{figure}

To help to understand, we explain the model using examples.
Figure \ref{fig:exconf} illustrates an example configuration by 8 robots; $R = \{r_1, r_2, \ldots, r_8\}$.
Robots $r_1$ and $r_2$ (resp. $r_4$, $r_5$ and $r_6$) are accompanied, and the other robots are single.
The dotted arrow between robots represents the distance between the points occupied by the robots.
Let $p_i$ denote the point occupied by robot $r_i$.
Now we explain the models as follows:

\begin{itemize}
    \item {\bf The adversarial (8,4)-defected model}. In this model, each robot observes 4 other robots chosen adversarially.
    Assume that robot $r_3$ observes 4 robots, $r_1$, $r_2$, $r_6$, and $r_8$. In this case, robot $r_3$ gets a set of points 
    $P^{r_3} = \{p_1^*, p_3, p_6, p_8 \}$ including point $p_3$ occupied by $r_3$ itself, 
    where $p_i^*$ denotes that $p_i$ is recognized to be occupied by two or more robots.
    Robot $r_3$ knows that two or more robots exist at point $p_1$ because both robots $r_1$ and $r_2$ are chosen, 
    however, $r_3$ does not know that there is another robot at $p_6$ because it observes only $r_6$ at $p_6$.
    Robot $r_4$ (or $r_5$, $r_6$) observes 4 robots among 5 robots, $r_1$, $r_2$, $r_3$, $r_7$, and $r_8$.
    If robots $r_1$, $r_3$, $r_7$, and $r_8$ are chosen, $P^{r_4} = \{p_1, p_3, p_4^*, p_7, p_8\}$ holds, which means that
    $r_4$ observes all points, however, it does not know that another robot exists at $p_1$.
    Robot $r_4$ can know that point $p_4$ is occupied by another robot other than itself.
    If robot $r_4$ observes robots $r_1$, $r_2$, $r_7$, and $r_8$, 
    $P^{r_4} = \{p_1^*, p_4^*, p_7, p_8\}$ holds, which means that
    robot $r_4$ knows there exist two or more robots at $p_1$, but it cannot observe point $p_3$ occupied by robot $r_3$.
    Notice that $r_5$ and $r_6$ located at $p_4$ are allowed to observe the set of points different from those observed by $r_4$.  
    \item {\bf The distance-based (8,3)-defected model}. In this model, each robot observes 3 closest robots to itself.
    Robot $r_7$ observes 3 robots, $r_8$ (the closest one) and two robots among three robots at point $p_4$, 
    thus $P^{r_7} = \{p_4^*, p_7, p_8\}$ always holds.
    Robot $r_4$ observes robot $r_3$ (the closest one)
    and two robots among 4 robots, $r_1$, $r_2$, $r_7$, and $r_8$, which are the same distance apart.
    Note that the observed robots are determined in an arbitrary way, 
    thus in this case, $P^{r_4}$ becomes one among $\{p_1^*, p_3, p_4^*\}$, $\{p_1, p_3, p_4^*, p_7\}$, $\{p_1, p_3, p_4^*, p_8\}$, 
    or $\{p_3, p_4^*, p_7, p_8\}$.
\end{itemize}

It is obvious that the adversarial ($N$,$k$)-defected model is weaker\footnote{Strictly speaking, we do not know the adversarial ($N$,$k$)-defected model is properly weaker than the distance-based one yet;
it is obvious that the adversarial ($N$,$k$)-defected model is NOT stronger
than the distance-based one.}
than the distance-based one, that is, any gathering algorithm for the adversarial ($N$,$k$)-defected model works correctly also in the distance-based ($N$,$k$)-defected model.

\subsection {Problem Definition: Gathering}
The \emph{Gathering} problem is to locate all robots at the same point.
We define the gathering problem as follows.

\begin{definition}
\textbf{The Gathering Problem} 
Given a set of $N$ robots located at arbitrary points.
Algorithm $\mathcal{A}$ solves the gathering problem 
if $\mathcal{A}$ satisfies all the following conditions:
(1) 
algorithm $\mathcal{A}$ eventually reaches 
a configuration such that no robot can move, 
and (2) when the algorithm $\mathcal{A}$ terminates, 
all the robots are located at the same point.
\end{definition}

\section{Algorithm in the Adversarial ($N$,$N-2$)-defected Model where $N \geq 5$}
\label{sec:nn-alg}
In this section, we introduce an algorithm to solve the gathering problem
in the adversarial ($N$, $N-2$)-defected model, where $N \geq 5$.

\begin{algorithm}[tbhp]
\caption{Gathering algorithm for robot $r_i$ in the adversarial $(N, N-2)$-defected model where $N\ge 5$}
\label{alg:nn-2gather}
{\bf functions:}
\begin{description}
\item[$~\cdot~ \opset()$] returns a set of points $\{p\:|\: p$ is occupied by $r_i$ or a point occupied by the robot that $r_i$ observed$\}$.
\item[$~\cdot~ \isMulti(p)$] returns $\TRUE$ if point $p$ is occupied by two or more robots that $r_i$ observed (weak multiplicity), otherwise $\FALSE$.
\end{description}
{\bf algorithm:}
\begin{algorithmic}[1]
\If{$\forall p \in \opset(): \isMulti(p) = \TRUE$}
\State{move to the center of the smallest enclosing circle of $\opset()$}
\ElsIf{($r_i$ is single) $\wedge$ $(\exists p \in \opset(): \isMulti(p) = \TRUE)$}
\State{move to an arbitrary point $p \in \opset()$ such that $\isMulti(p) = \TRUE$}
\ElsIf{$\forall p \in \opset(): \isMulti(p) = \FALSE$}
\State{move to the center of the smallest enclosing circle of $\opset()$}
\EndIf \Comment{No action if ($r_i$ is accompanied) $\wedge$ $(\exists p \in \opset(): \isMulti(p) = \FALSE)$}
\end{algorithmic}
\end{algorithm}

Algorithm \ref{alg:nn-2gather} presents an algorithm to achieve the gathering in the adversarial ($N$, $N-2$)-defected model where $N \ge 5$.
%
The algorithm adopts, as the destination of robot $r_i$, the center of the smallest enclosing circle (SEC) of the occupied points that $r_i$ observed in the \LL\ operation.
Before proving the correctness of the algorithm, we show some fundamental properties of the SEC of points in a plane.

\begin{proposition}
\label{prop:SEC}
Let $P$ be a set of $n$ distinct points in a plane and $C$ be the SEC of $P$.
The following properties hold.
\begin{enumerate}
    \item The SEC of $P$ is unique.
    \item Let $p \in P$ be any point (if exists) properly inside $C$, $C$ is the SEC of $P \setminus \{p\}$.
    \item When there exist three points $p_1, p_2, p_3 \in P$ on the boundary of $C$ that form an acute or right triangle, $C$ is the SEC of $\{p_1, p_2, p_3\}$.
    \item When three or more points in $P$ are on the boundary of $C$, there exist three points $p_1, p_2, p_3 \in P$ on the boundary of $C$ that form an acute or right triangle. \qed
  \end{enumerate}
\end{proposition}

A key property of the ($N$, $N-2$)-defected model, 
which is used in the following proofs, 
is that \emph{any accompanied robot can observe all the robots} (but only with the weak multiplicity detection).

\begin{lemma}
\label{lemma:3ac}
In the adversarial ($N$, $N-2$)-defected model ($N \ge 5$), Algorithm \ref{alg:nn-2gather} solves the gathering problem in two rounds from any configuration where there exist three or more accompanied robots.
\end{lemma}

\begin{proof}
When every robot is accompanied, each robot detects all the occupied points in the \LL\ operation and recognizes that each of them is occupied by multiple robots.  Every robot moves to the center of the SEC of all the occupied points (by lines 1 and 2 in Algorithm \ref{alg:nn-2gather}) and thus the gathering is achieved in one round.

When there exists a single robot $r$, every accompanied robot observes $r$ and does not move (see line 7 in Algorithm \ref{alg:nn-2gather}).  Every single robot misses at most one accompanied robot in its \LL\ operation and can detect at least one point occupied by multiple robots: a point occupied by three or more robots (if exists) or one of the two points each occupied by two robots.  Each single robot moves to one of such points (by lines 3 and 4 in Algorithm \ref{alg:nn-2gather}), which results in the configuration where every robot is accompanied.  Thus the gathering is achieved in the next round as shown above. 
\end{proof}

\begin{lemma}
\label{lemma:2ac}
In the adversarial ($N$, $N-2$)-defected model ($N \ge 5$), Algorithm \ref{alg:nn-2gather} solves the gathering problem in two rounds from any configuration where there exist only two accompanied robots (that are at the same point).
\end{lemma}

\begin{proof}
Let $r_1$ and $r_2$ be the two accompanied robots.  Robots $r_1$ and $r_2$ observe all robots and recognize that single robots exist, which makes $r_1$ and $r_2$ stay at the current point.

Now consider actions of single robots.
A single robot $r$ misses one robot in its \LL\ operation, which implies that $r$ observes (a) both $r_1$ and $r_2$ or (b) only one of $r_1$ and $r_2$.
In case (a), $r$ moves to the point, say $p_a$, occupied by $r_1$ and $r_2$.  
In case (b), $r$ moves to the center, say $p_b$, of the SEC of all the occupied points.
Thus after one round, all the robots are located at $p_a$ or $p_b$.  Note that $p_a$ is occupied by multiple robots including $r_1$ and $r_2$.

When $p_b$ is not occupied by any robot, the gathering is already achieved.
When $p_b$ is occupied by multiple robots, the robots at $p_b$ observe all the robots. Thus, all the robots move to the center of the SEC of $p_a$ and $p_b$ (or the midpoint of $p_a$ and $p_b$) in the next round (by lines 1 and 2 in Algorithm \ref{alg:nn-2gather}), which achieves the gathering.
When $p_b$ is occupied by only one robot $r$, $r$ detects that $p_a$ is occupied by multiple robots and moves to $p_a$ in the next round (by lines 3 and 4 in Algorithm \ref{alg:nn-2gather}) while the robots at $p_a$ recognize that $p_b$ is occupied by only one robot and does not move (see line 7 in Algorithm \ref{alg:nn-2gather}).  Thus, the gathering is achieved.   
\end{proof}

\begin{lemma}
\label{lemma:0ac}
In the adversarial ($N$, $N-2$)-defected model ($N \ge 5$), Algorithm \ref{alg:nn-2gather} solves the gathering problem in three rounds from any configuration where all robots are single.
\end{lemma}

\begin{proof}
Each robot misses one robot in its \LL\ operation.
When there exist two robot $r_1$ and $r_2$ that miss the same robot, $r_1$ and $r_2$ get the same point set $\opset()$ and moves to the center of the SEC of $\opset()$ (by lines 5 and 6 in Algorithm \ref{alg:nn-2gather}).  
From Lemmas \ref{lemma:3ac} and \ref{lemma:2ac}, two additional rounds are enough to achieve the gathering.

When no two robots miss the same robot, for any pair of two distinct robots $r_1$ and $r_2$, the robot missing $r_1$ is different from the robot missing $r_2$.  Let $C$ be the SEC of all the occupied $N$ points.  First, consider the case that two (or more) robots $r_a$ and $r_b$ are located properly inside $C$.  
The SEC of $R\setminus \{r_a\}$ is also the SEC of $R\setminus \{r_b\}$ (that is $C$ from the second property of Proposition \ref{prop:SEC}), which implies that the two robots observing $R\setminus \{r_a\}$ and $R\setminus \{r_b\}$ move to the same point (or the center of the SEC).  From Lemmas \ref{lemma:3ac} and \ref{lemma:2ac}, two additional rounds are enough to achieve the gathering.

Second, consider the case that $N-1$ or $N$ robots are on the boundary of $C$.  From the last property of Proposition \ref{prop:SEC}, there exist three robots $r_1, r_2, r_3$ on the boundary of $C$ that form an acute or right triangle.  There exist two robots $r_4$ and $r_5$ other than $r_1, r_2, r_3$ from $N \ge 5$.  Both the robots observing $R\setminus \{r_4\}$ and $R\setminus \{r_5\}$ observe all of $r_1, r_2, r_3$. The third property of Proposition \ref{prop:SEC} implies that the two robots find the same SEC of $\opset()$ (or the SEC of $r_1, r_2, r_3$), which implies that they move to the same point (or the center of the SEC) (by lines 5 and 6 in Algorithm \ref{alg:nn-2gather}).
From Lemmas \ref{lemma:3ac} and \ref{lemma:2ac}, two additional rounds are enough to achieve the gathering.
\end{proof}

From Lemmas \ref{lemma:3ac}, \ref{lemma:2ac} and \ref{lemma:0ac}, the following theorem holds.

\begin{theorem}
In the adversarial ($N$, $N-2$)-defected model ($N \ge 5$), Algorithm \ref{alg:nn-2gather} solves the gathering problem in three rounds. \qed
\end{theorem}

Algorithm \ref{alg:nn-2gather} cannot solve the gathering problem for the case of $N=4$.
Assume that four robots, $R=\{r_0, r_1, r_2, r_3\}$.  Three robots $r_1, r_2$ and $r_3$ are deployed to form an equilateral triangle as Figure \ref{impos_eqtri} and $r_0$ is located at the center of the triangle (\ie point $p_c$ in Figure \ref{impos_eqtri}).
Consider the case that $r_i$ observes $r_{(i+1)\ \text{mod}\ 3}$ and $r_{(i+2)\ \text{mod}\ 3}$ for each $i\ (0 \le i \le 3)$.
According to Algorithm \ref{alg:nn-2gather}, $r_0$ moves to the midpoint of $r_1$ and $r_2$, $r_1$ moves to $p_1$, $r_2$ moves to the midpoint of $r_2$ and $r_3$, and $r_3$ moves to the midpoint of $r_3$ and $r_1$.  
In the resultant configuration, $r_0, r_2$ and $r_3$ form an equilateral triangle and $r_1$ is located at the center $p_1$ of the triangle, which shows by repeating the argument that the gathering is never achieved.

Thus we need another gathering algorithm for the adversarial (4, 2)-defected model, 
however, we do not know whether the gathering problem in the adversarial (4,2)-defected model is solvable or not yet.
In the next section, we present an algorithm to solve the gathering problem 
in the distance-based (4,2)-defected model.

\section{Algorithm in the Distance-based (4,2)-defected Model}
\label{sec:42alg}
Now we present an algorithm to solve the gathering problem
from any arbitrary initial configuration 
in the distance-based (4,2)-defected model.

\noindent
{\bf Basic Strategy}.
In this model, the number of robots is 4 and each robot observes at most two other points occupied by some other robots 
(three points in total including one occupied by itself).
In other words, the observation result of each robot 
forms a triangle (by three points/robots) when every robot is single.
The strategy of the proposed algorithm is 
to determine one unique point from the formed triangle.
Therefore, if any two robots observe the same three points occupied by the robots (including itself), 
the two robots move to the same point according to the proposed algorithm.
If two or more robots are accompanied, the gathering can be achieved as the same manner 
introduced in Algorithm \ref{alg:nn-2gather}.
Obviously, in this strategy, we have to consider the case so that all 4 robots observe the different triangles.
We resolve this problem by the geometrical property 
(remind that each robot cannot observe the farthest robot from itself in 
the distance-based defected model).

\begin{algorithm}[bthp]
\caption{Gathering algorithm for robot $r_i$ in the distance-based (4,2)-defected model}
\label{alg:42gather}
{\bf functions:}
\begin{description}
\item[$~\cdot~ \opset()$] returns a set of points $\{p\:|\: p$ is occupied by $r_i$ or a point occupied by the robot that $r_i$ observed$\}$.
\item[$~\cdot~ \isMulti(p)$] returns $\TRUE$ if point $p$ is occupied by two or more robots that $r_i$ observed (weak multiplicity), otherwise $\FALSE$.
\end{description}
{\bf algorithm:}
\begin{algorithmic}[1]
\If{$\forall p \in \opset(): \isMulti(p) = \TRUE$}
\State{move to the center of the smallest enclosing circle of $\opset()$}
\ElsIf{($r_i$ is single) $\wedge$ $(\exists p \in \opset(): \isMulti(p) = \TRUE)$}
\State{move to an arbitrary point $p \in \opset()$ such that $\isMulti(p) = \TRUE$}
\ElsIf{$\forall p \in \opset(): \isMulti(p) = \FALSE$}
    \If{$\opset()$~forms an equilateral triangle}
    \State{move to the center of the triangle (\ie incenter)} \Comment{Rule 1}
    \ElsIf{$\opset()$~forms an isosceles triangle}
    \State{move to the midpoint of the base of the triangle} \Comment{Rule 2}
    \Else \Comment{the other triangle or collinear three points}
    \State{move to the midpoint of the longest line} \Comment{Rule 3}
    \EndIf 
\EndIf \Comment{No action if ($r_i$ is accompanied) $\wedge$ $(\exists p \in \opset(): \isMulti(p) = \FALSE)$}
\end{algorithmic}
\end{algorithm}

Algorithm \ref{alg:42gather} presents an algorithm to achieve the gathering in the distance-based (4,2)-defected model.
Each robot which does not observe any accompanied robots executes one among three rules
(lines from 6 to 11 in Algorithm \ref{alg:42gather}). 
Figure \ref{alg42rules} illustrates these three rules.
If a robot observes an equilateral triangle (\ie the points observed by the robot form an equilateral triangle), 
it moves to the center of the triangle (Figure \ref{alg42rules}(a)), 
and if it observes an isosceles triangle, it moves to the midpoint of the base of the triangle (Figure \ref{alg42rules}(b)).
In the other case, it moves to the center point of the longest line of the triangle (Figure \ref{alg42rules}(c)).
It is obvious that if any two robots observe the same set of points (\ie the same \emph{view}: $r_i.\opset() = r_j.\opset()$, where $i \neq j$),
the two robots move to the same point according to Algorithm \ref{alg:42gather}.
Hence the following lemma holds.


\begin{lemma}
\label{lem:sameview}
In any configuration, if two or more robots have the same view, 
the robots move to the same point in one round by Algorithm \ref{alg:42gather}. \qed
\end{lemma}

\begin{figure}[b]
	\begin{center}
		\subfloat[Case of an equilateral triangle]{\includegraphics[scale=0.7]{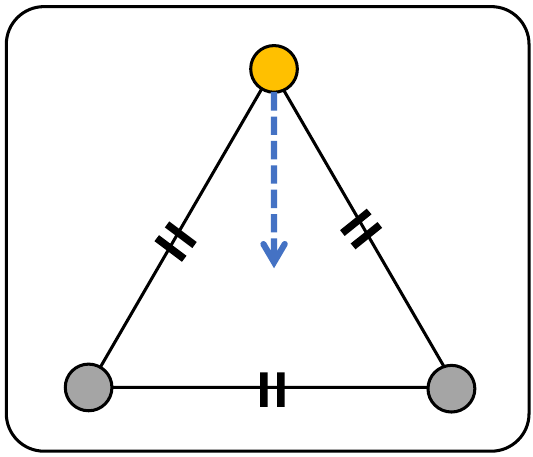}}
        \hspace{10pt}
		\subfloat[Case of an isosceles triangle]{\includegraphics[scale=0.7]{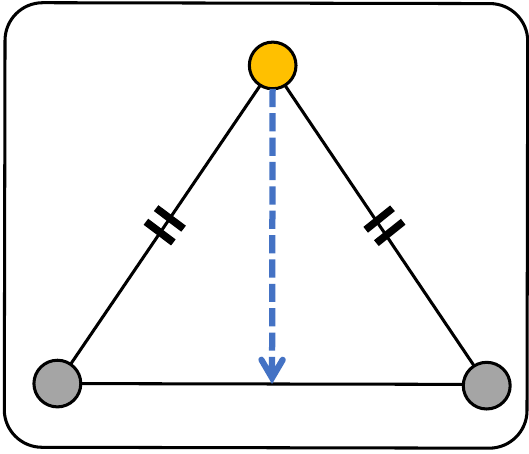}}
        \hspace{10pt}
		\subfloat[The other case]{\includegraphics[scale=0.7]{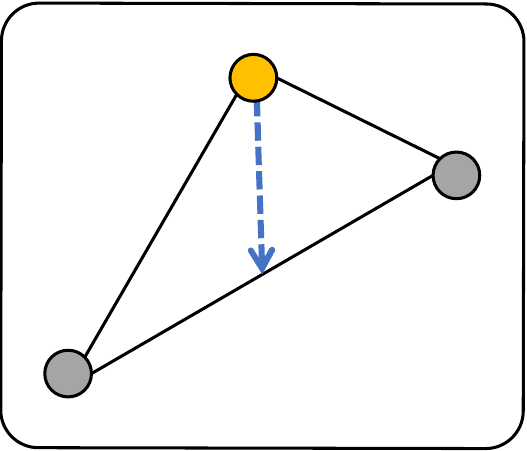}}
	\end{center}
    \caption{Three rules in Algorithm \ref{alg:42gather}}
    \label{alg42rules}
\end{figure}

In Algorithm \ref{alg:42gather}, actions when a robot observes any accompanied robots (including itself) 
are the exactly same as Algorithm \ref{alg:nn-2gather} (lines from 1 to 4 in both algorithms).
Therefore, by Lemma \ref{lemma:3ac}, the gathering is achieved in one round when there are three or more accompanied robots.
Now we consider the case that there are two accompanied robots 
in the following lemma.

\begin{lemma}
\label{lemma:2acin4}
In the distance-based (4,2)-defected model, 
Algorithm \ref{alg:42gather} solves the gathering problem in two rounds 
from any configuration where there exist only two accompanied robots.
\end{lemma}

\begin{proof}
Let $r_1$ and $r_2$ be the two accompanied robots, and let $r_3$ and $r_4$ be the two single robots. 
Robots $r_1$ and $r_2$ observe all robots and recognize that single robots exist, which makes $r_1$ and $r_2$ stay at the current point 
(by line 13).
If robot $r_3$ (or $r_4$) observes both $r_1$ and $r_2$, it moves to the point of $r_1$ and $r_2$ (by lines 3 and 4), 
this results in the configuration where there exist three or more accompanied robots and the gathering can be achieved by Lemma \ref{lemma:3ac}.
Now we assume that two robots $r_3$ and $r_4$ miss either $r_1$ or $r_2$ (\ie they observe only one robot either $r_1$ or $r_2$).
This implies that robots $r_3$ and $r_4$ observe each other, 
and these two robots observe the exactly same set of (three) points because robots $r_1$ and $r_2$ are at the same point (by lines from 6 to 11).
Therefore, robots $r_3$ and $r_4$ move to the same point in one round (by Lemma \ref{lem:sameview}) and the gathering is achieved in the next round.
\end{proof}

Even when all 4 robots are single, if two or more robots observe the same set of points, 
the robots move to the same point (by Lemma \ref{lem:sameview}), 
thus the gathering is achieved by Lemma \ref{lemma:3ac} and Lemma \ref{lemma:2acin4}.

Now we show that the gathering is eventually achieved in any configuration 
where all 4 (single) robots have the different views (\ie observe the different set of points).

\begin{lemma}
\label{lem:convex}
In the distance-based (4,2)-defected model, 
if all robots have the different views, 
the shape formed by the robots is a convex rectangle.
\end{lemma}

\begin{proof}
We prove the contraposition of the lemma: 
if the robots do not form a convex rectangle, there exist two robots having the same view.

\begin{figure}[t]
	\begin{center}
	    \includegraphics[scale=0.7]{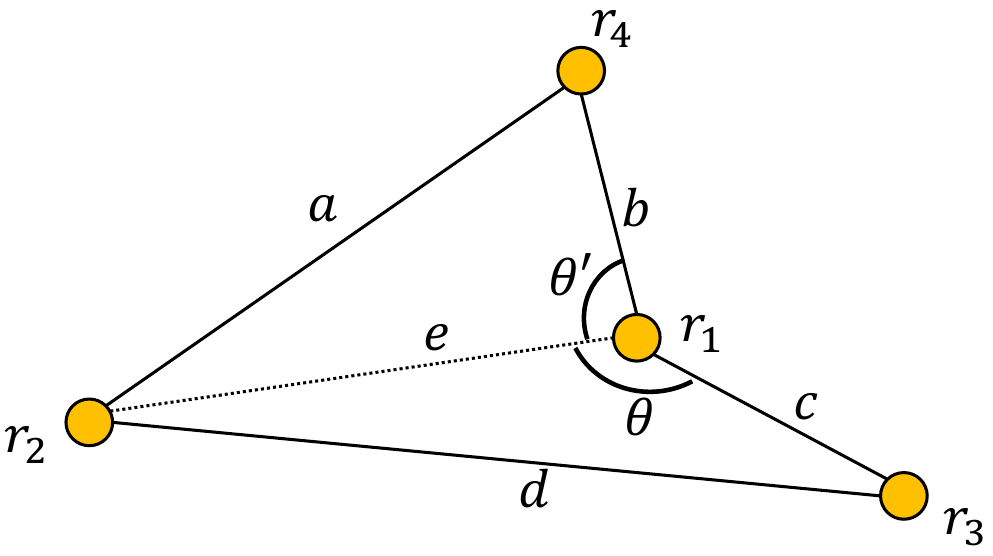}
	\end{center}
    \caption{An example of a concave rectangle}
    \label{concave}
\end{figure}

Assume that the 4 robots, from $r_1$ to $r_4$, form a concave rectangle as Figure \ref{concave}
(Note that we can also assume that the robots form a triangle (\ie three robots are collinear), it can be also proved in the same manner).
A concave rectangle has an interior angle which is larger than $180^{\circ}$, 
so we assume robot $r_1$ is located at the point with such an angle as Figure \ref{concave}.
Let $e$ be the line $\overline{r_1r_2}$, either angle $\angle{r_2r_1r_4}$ or angle $\angle{r_2r_1r_3}$ is an obtuse angle (\ie angle larger than $90^{\circ}$) 
because interior angle $\angle{r_4r_1r_3}$ is larger than $180^{\circ}$.
Without loss of generality, we assume angle $\angle{r_2r_1r_3}$ is an obtuse angle (denoted by $\theta$).
Due to $\theta > 90^{\circ}$, $d$ is longer than $c$ and $e$ (see Figure \ref{concave}). 
This implies that robot $r_3$ observes $r_1$ and robot $r_2$ also observes $r_1$ (because the farthest robot is missed in the distance-based defected model).
If angle $\angle{r_2r_1r_4}$ (denoted by $\theta'$) is also an obtuse angle, robot $r_4$ also observes $r_1$ by the same reason.
As a result, all robots observe $r_1$ (including $r_1$ itself) and the lemma holds 
because there is a robot which has the same view with robot $r_1$ by the pigeonhole principle.
If angle $\angle{r_2r_1r_4}$ is an acute angle (\ie angle smaller than $90^{\circ}$) or a right angle, 
$\theta + \theta' < 270^\circ$ holds.
This means that an exterior angle $\angle{r_4r_1r_3}$ (\ie $360^\circ - \theta - \theta'$) is an obtuse angle, 
thus $b$ is shorter than $\overline{r_4r_3}$. Also in this case, robot $r_4$ observes $r_1$ and the lemma holds.
\end{proof}

\begin{lemma}
\label{lem:unob22}
Assume that all robots have the different views.
If robot $r_i$ cannot observe robot $r_j$ (\ie robot $r_i$'s view does not include the point occupied by $r_j$), 
$r_j$ cannot observe $r_i$ neither.
\end{lemma}

\begin{proof}
To help to explain, we introduce a directed graph $\vec{G} = (V,A)$ such that 
$V = \{r_1, r_2, r_3, r_4\}$ and $(r_i, r_j) \in A$ if robot $r_i$ cannot observe $r_j$.
If all robots have the different views, there exist only two cases as Figure \ref{dag}.
And we show that there is no case as Figure \ref{dag}(a) to prove the lemma.


\begin{figure}[b]
 \begin{minipage}{0.6\hsize}
	\begin{center}
		\subfloat[]{\includegraphics[scale=0.7]{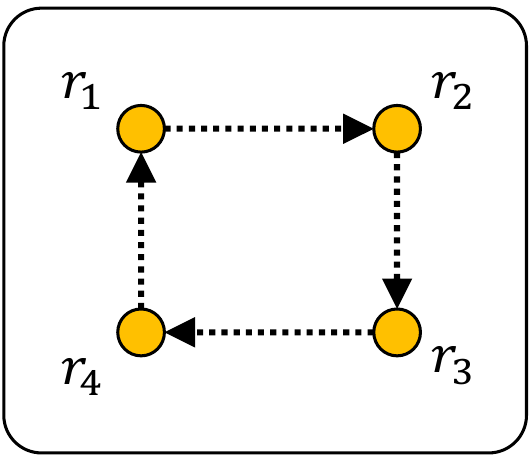}}
		\subfloat[]{\includegraphics[scale=0.7]{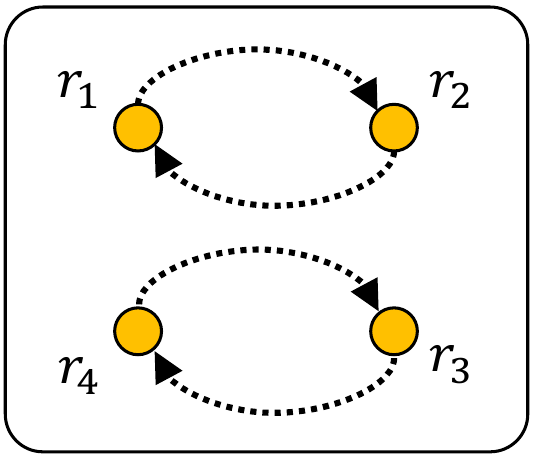}}
	\end{center}
    \caption{Directed graphs showing an unobserved robot}
    \label{dag}
 \end{minipage}
 \hfill
 \begin{minipage}{0.4\hsize}
	\begin{center}
	    \includegraphics[scale=0.7]{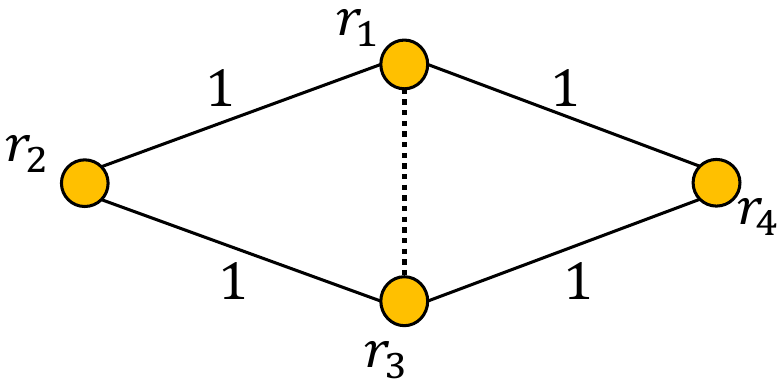}
	\end{center}
    \caption{Two isosceles triangles}
    \label{twotri}
 \end{minipage}
\end{figure}

Assume the case as Figure \ref{dag}(a): 
robot $r_1$ cannot observe $r_2$, robot $r_2$ cannot observe $r_3$, and so on.
$\overline{r_1r_4} \leq \overline{r_1r_2}$ holds because robot $r_1$ cannot observe $r_2$.
For the same reason, 
$\overline{r_1r_2} \leq \overline{r_2r_3}$, $\overline{r_2r_3} \leq \overline{r_3r_4}$,
and $\overline{r_3r_4} \leq \overline{r_1r_2}$ also hold.
Therefore, $\overline{r_1r_4} \leq \overline{r_1r_2} \leq \overline{r_2r_3} \leq \overline{r_3r_4} \leq \overline{r_1r_4}$ holds, 
thus $\overline{r_1r_2} = \overline{r_2r_3} = \overline{r_3r_4} = \overline{r_1r_4}$ holds.
For simplicity, we assume that the length of $\overline{r_1r_2}$ is 1.

Now we consider the triangle $\triangle{r_1r_2r_3}$. 
Due to $\overline{r_1r_2} = \overline{r_2r_3}$, triangle $\triangle{r_1r_2r_3}$ is an isosceles triangle (the base is $\overline{r_1r_3}$).
Similarly, triangle $\triangle{r_1r_3r_4}$ is also an isosceles triangle which has line $\overline{r_1r_3}$ as the base.
Line $\overline{r_1r_3}$ is the common base of these two isosceles triangles, 
thus the locations of 4 robots are as Figure \ref{twotri}.


In Figure \ref{twotri}, we consider the lengths of two diagonal lines, $\overline{r_1r_3}$ and $\overline{r_2r_4}$.
By the assumption, robot $r_1$ cannot observe $r_2$, therefore, $\overline{r_1r_3} \leq 1$ holds because robot $r_1$ observes $r_3$.
As the same reason, $\overline{r_2r_4} \leq 1$ also holds.
However, both $\overline{r_1r_3} \leq 1$ and $\overline{r_2r_4} \leq 1$ cannot hold in this rhombus,
therefore, there is no case as Figure \ref{dag}(a) and the lemma holds.
\end{proof}

By Lemma \ref{lem:unob22}, if all robots have different views, 
we have two disjoint pairs of robots such that robots in each pair cannot observe each other as in Figure \ref{dag}(b).
Now we discuss the location relations among the robots in this case by the following lemma.

\begin{lemma}
\label{lem:relation}
If all robots have different views in the distance-based (4,2)-defected model,
each of two robots which cannot be observed each other are 
diagonally located on the formed convex rectangle.
\end{lemma}

\begin{proof}
We already proved that the robots form a convex rectangle if all robots have different views by Lemma \ref{lem:convex}.
Let $r_1$ and $r_2$ be two robots which do not observe each other, 
and we assume for contradiction that $r_1$ and $r_2$ are not diagonally located (\ie line $\overline{r_1r_2}$ is an edge of the convex rectangle).
For simplicity, we assume the length of $\overline{r_1r_2}$ is 1.

\begin{figure}[t]
	\begin{center}
	    \includegraphics[scale=0.5]{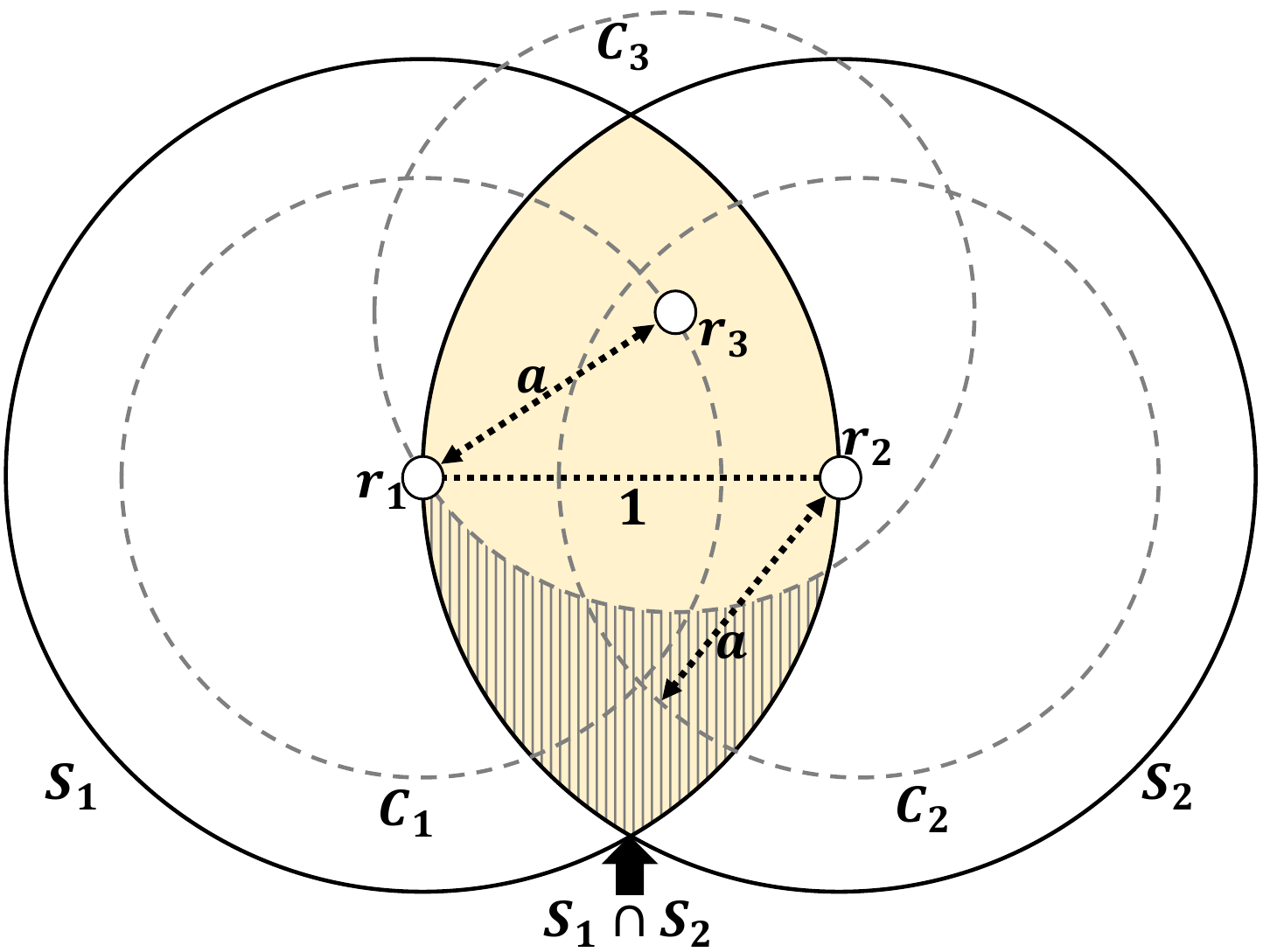}
	\end{center}
    \caption{Possible positions of robots $r_3$ and $r_4$}
    \label{twocircle1}
\end{figure}

Figure \ref{twocircle1} illustrates two circles, called $S_1$ and $S_2$, with radius 1
centered at $r_1$ and $r_2$ respectively.
Consider the position of robot $r_3$: robot $r_3$ should be located in area $S_1 \cap S_2$,
because both $r_1$ and $r_2$ observe $r_3$ (remind that $r_1$ and $r_2$ do not observe each other).
Locate $r_3$ in an arbitrary point in area $S_1 \cap S_2$. 
Let $a$ be the length of longer one between $\overline{r_1r_3}$ and $\overline{r_2r_3}$ (\ie $a = max(|\overline{r_1r_3}|,|\overline{r_2r_3}|)$), 
here we assume $a$ is the length of $\overline{r_1r_3}$ without loss of generality.
Circles $C_1$, $C_2$, and $C_3$ present the circles with radius $a$ centered at $r_1$, $r_2$, and $r_3$ respectively.
By Lemma \ref{lem:unob22}, robots $r_3$ and $r_4$ cannot observe each other, 
thus $|\overline{r_3r_4}| \geq a$ holds; robot $r_4$ should be located outside of $C_3$.
As a result, robot $r_4$ should be located in $(C_1 \cap C_2) - C_3$ 
which is presented as the shaded area in Figure \ref{twocircle1}.
In this case, robots $r_1$ and $r_2$ (resp. $r_3$ and $r_4$) are diagonally located on a convex rectangle,
which is a contradiction. Thus the lemma holds.
\end{proof}

Now we show that even when all single robots have different views, 
two or more robots move to the same point by Algorithm \ref{alg:42gather}.
We consider the 6 lines derived by the combination of 4 robots (refer to Figure \ref{fig:6lines}).
We focus on the lengths of these 6 lines, 
and the following corollary holds by Lemma \ref{lem:relation}.


\begin{figure}[b]
 \begin{minipage}{0.5\hsize}
  \begin{center}
   \includegraphics[scale=0.55]{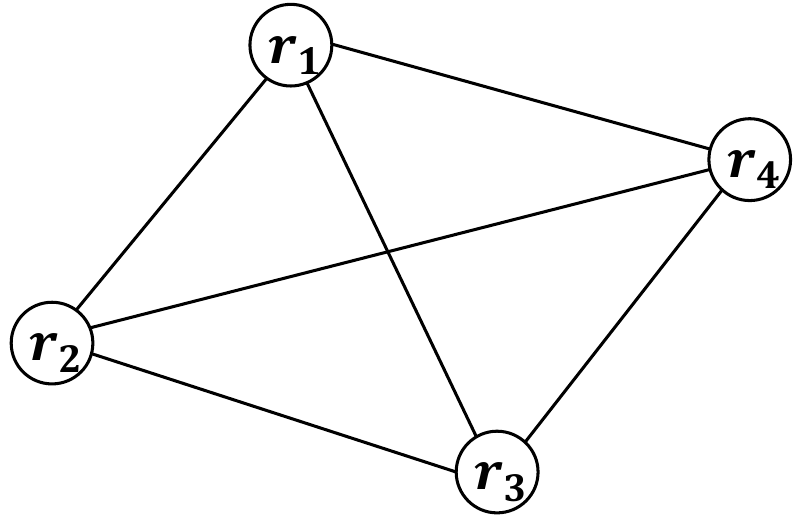}
  \end{center}
  \caption{6 lines by 4 robots}
  \label{fig:6lines}
 \end{minipage}
 \hfill
 \begin{minipage}{0.5\hsize}
	\begin{center}
	    \includegraphics[scale=0.55]{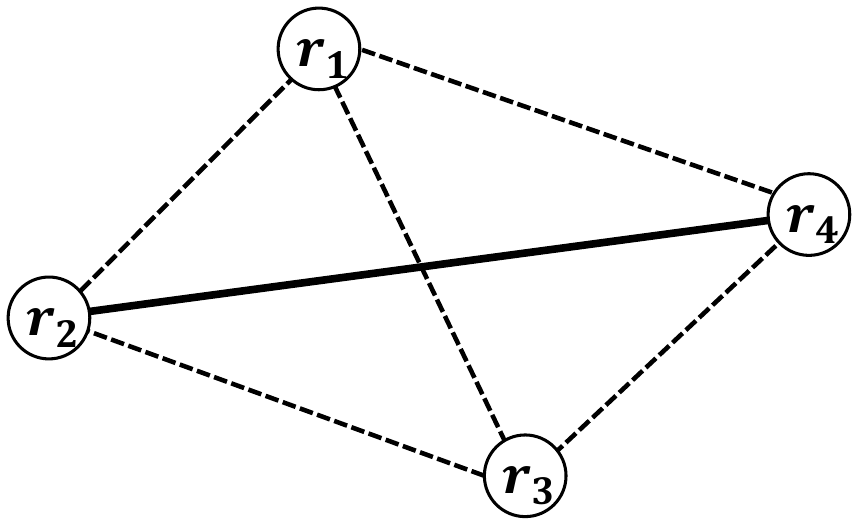}
	\end{center}
    \caption{Configuration with one longest line}
    \label{fig:1long}
 \end{minipage}
\end{figure}

\begin{corollary}
\label{col:longest}
Consider the 6 lines connecting distinct pairs of two robots.
If all robots are single and have different views, 
there is no (side) line which is longer than any diagonal line.
\end{corollary}

It is worthwhile to mention that there can be at most 4 longest lines among 6 lines.
We focus on the number of longest lines and show that the algorithm works correctly in all cases.
By Corollary \ref{col:longest}, if there exist one or two longest lines, they are diagonal lines.
The following lemma holds.

\begin{lemma}
\label{lem:1long}
Assume that all robots are single and have different views in the distance-based (4,2)-defected model, 
and consider the 6 lines connecting distinct pairs of two robots.
If there exist one or two longest lines, two or more robots become accompanied in one round.
\end{lemma}


\begin{proof}
Figure \ref{fig:1long} illustrates an example configuration including the only one longest line (as a diagonal line), 
where the thick solid line represents the unique longest line.
Without loss of generality, we assume that line $\overline{r_2r_4}$ is the longest one.
From the assumption, $r_1$ and $r_3$ do not observe each other: 
$r_1$ observes triangle $\triangle{r_1r_2r_4}$, and $r_3$ observes triangle $\triangle{r_2r_3r_4}$.
These two triangles are not equilateral triangles because line $\overline{r_2r_4}$ is the unique longest line.
Therefore, robots $r_1$ and $r_3$ move to the midpoint of line $\overline{r_2r_4}$ (by line 9 or 11).
If there are two longest lines, the both lines are diagonal lines by Corollary \ref{col:longest} 
($\overline{r_1r_3}$ and $\overline{r_2r_4}$ in Figure \ref{fig:1long}).
However, this does not affect to the actions of robots $r_1$ and $r_3$; they move to the midpoint of line $\overline{r_2r_3}$.
Thus the lemma holds.
\end{proof}


Now we consider the case that 
there is a side line whose length is the same as two diagonal lines;
there are three longest lines.

\begin{lemma}
\label{lem:3long}
Assume that all robots are single and have different views in the distance-based (4,2)-defected model, 
and consider the 6 lines connecting distinct pairs of two robots.
If there are the three longest lines, two or more robots become accompanied in two rounds.
\end{lemma}


\begin{proof}
Figure \ref{fig:3long} illustrates the only configuration including three longest lines.
Three thick solid lines are the three longest lines.
Remind that robots $r_1$ and $r_3$ (or $r_2$ and $r_4$) cannot observe each other.
By Algorithm \ref{alg:42gather}, 
all robots move to the different points:
robot $r_1$ (resp. $r_2$) moves to the midpoint $p_1$ (resp. $p_2$) of line $\overline{r_1r_4}$ (resp. $\overline{r_2r_3}$) since $r_1$ (resp. $r_2$) observes an isosceles triangle $\triangle{r_1r_2r_4}$ (resp. $\triangle{r_1r_2r_3}$).
Robot $r_3$ (resp. $r_4$) moves to the midpoint $p_3$ (resp. $p_4$) of line $\overline{r_2r_4}$ (resp. $\overline{r_1r_3}$) that is the longest line of the observed triangle $\triangle{r_2r_3r_4}$ (resp. $\triangle{r_1r_3r_4}$).
In this case, 
triangles $\triangle{r_2r_3r_4}$ and $\triangle{r_2p_2p_3}$ are similar, 
the length of line $\overline{p_2p_3}$ is half of the length of line $\overline{r_3r_4}$, 
and line $\overline{p_2p_3}$ and line $\overline{r_3r_4}$ are parallel.
Through the same argument for lines $\overline{p_1p_4}$ and $\overline{r_4r_3}$, we can show that the lengths of lines $\overline{p_1p_4}$ and  $\overline{p_2p_3}$ are the same
and these two lines are parallel.
This means that the rectangle formed in the next round is a parallelogram:
even if all robots have different views in this configuration,
two or more robots become accompanied in the next round
because diagonal line $\overline{p_1p_2}$ is the unique longest line (by Lemma \ref{lem:1long}).
\end{proof}

\begin{figure}[t]
 \begin{minipage}{0.5\hsize}
	\begin{center}
	    \includegraphics[scale=0.5]{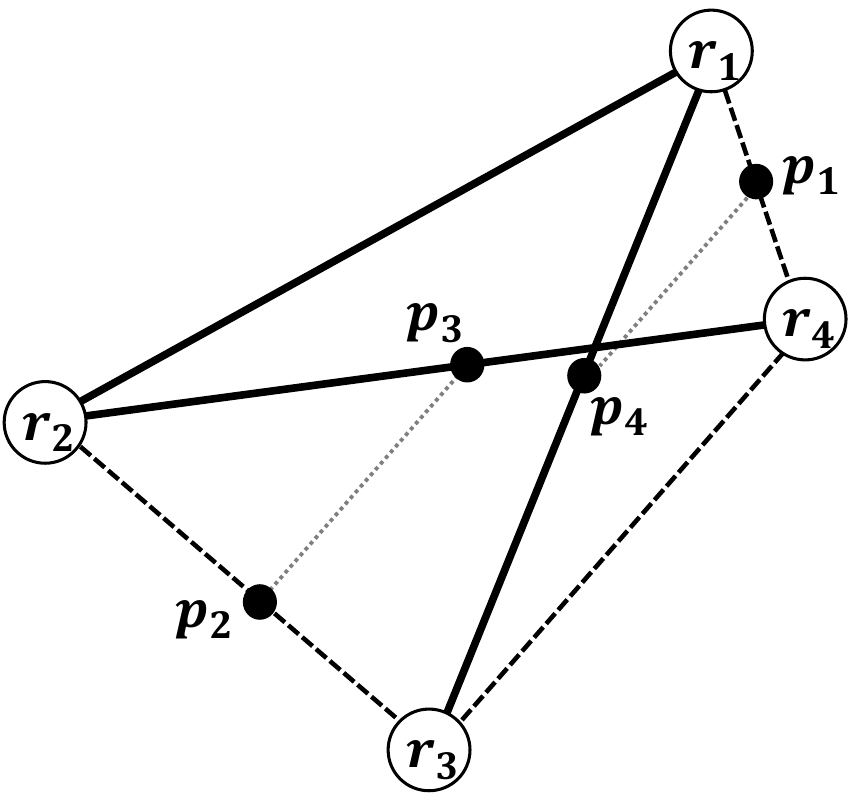}
	\end{center}
    \caption{Configuration with three longest lines}
    \label{fig:3long}
 \end{minipage}
 \hfill
 \begin{minipage}{0.5\hsize}
	\begin{center}
	    \includegraphics[scale=0.46]{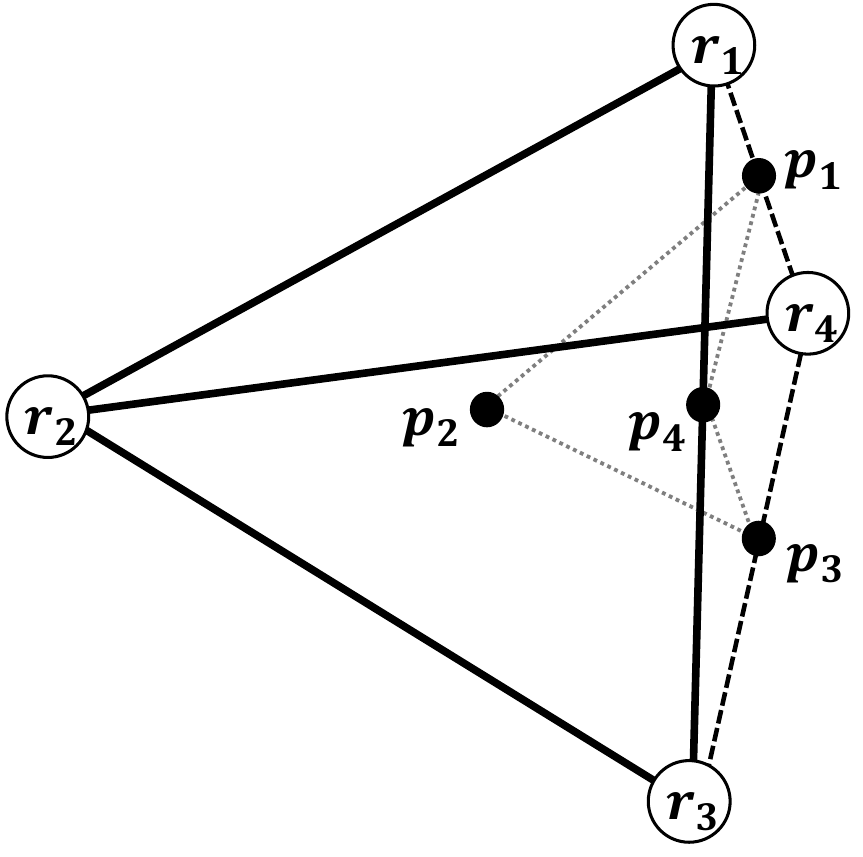}
	\end{center}
    \caption{Configuration with four longest lines}
    \label{fig:4long}
 \end{minipage}
\end{figure}

\begin{lemma}
\label{lem:4long}
Assume that all robots are single and have different views in the distance-based (4,2)-defected model, 
and consider the 6 lines connecting distinct pairs of two robots.
If there are four longest lines, two or more robots become accompanied in two rounds.
\end{lemma}


\begin{proof}
Figure \ref{fig:4long} illustrates the only configuration including four longest lines.
Four thick solid lines are the four longest lines.
By Algorithm \ref{alg:42gather}, 
all robots move to the different points:
robot $r_1$ (resp. $r_3$) moves to the midpoint $p_1$ (resp. $p_3$) of line $\overline{r_1r_4}$ (resp. $\overline{r_3r_4}$) since $r_1$ (resp. $r_3$) observes an isosceles triangle $\triangle{r_1r_2r_4}$ (resp. $\triangle{r_2r_3r_4}$). 
Robot $r_2$ moves to the center $p_2$ of the equilateral triangle $\triangle{r_1r_2r_3}$ it observes,  
and $r_4$ moves to the midpoint $p_4$ of the unique longest line $\overline{r_1r_3}$ it observes 
(note that if $|\overline{r_1r_4}|=|\overline{r_3r_4}|$, triangle $\triangle{r_1r_3r_4}$ is an isosceles triangle, 
however robot $r_4$ moves to the midpoint $p_4$ of the base line also in this case).
As a result, the four points, from $p_1$ to $p_4$, form a concave rectangle.
Hence, two or more robots become accompanied in the next round by Lemma \ref{lem:convex}.
\end{proof}

From Lemmas \ref{lemma:3ac}, \ref{lemma:2acin4}, \ref{lem:1long}, \ref{lem:3long} and \ref{lem:4long}, 
the following theorem holds.

\begin{theorem}
In the distance-based (4, 2)-defected model, Algorithm \ref{alg:42gather} solves the gathering problem in at most four rounds. \qed
\end{theorem}

\section{Impossibility Results}
\label{sec:impos}
In this section, we present two impossibility results for the gathering problem in the defected view model;
(1) there is no (deterministic) algorithm in the arbitrary or distance-based (3,1)-defected model, 
and 
(2) there is no (deterministic) algorithm in the relaxed adversarial ($N$,$N-2$)-defected model defined in Section \ref{sec:relaxed}.

\subsection{Impossibility in (3,1)-defected model}
By two gathering algorithms we introduced in the previous sections, 
the gathering can be achieved in the adversarial (and thus also in the distance-based) 
($N$,$N-2$)-defected model for $N \geq 5$, 
and in the distance-based (4,2)-defected model.
These results bring us a problem to find an algorithm to solve the gathering problem 
in the distance-based (or adversarial) (3,1)-defected model.
Here we show that there is no such algorithm.

\begin{theorem}
There is no (deterministic) algorithm to solve the gathering problem 
in the distance-based or adversarial (3,1)-defected model.
\end{theorem}

\begin{figure}[t]
  \centering
		\includegraphics[scale=0.8]{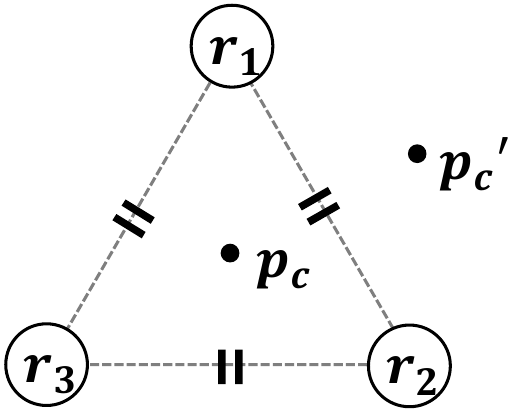}
  \caption{Example for an unsolvable configuration in the distance-based (3,1)-defected model}
  \label{impos_eqtri}
\end{figure}

\begin{proof}
We prove this theorem by showing that there is no (deterministic) algorithm even in the distance-based (3,1)-defected model.
Note that the distance-based (3,1)-defected model is stronger than the adversarial one, 
hence if the gathering cannot be achieved in the distance-based one, the gathering is also unsolvable in the adversarial one.
Assume that three robots, $R=\{r_1, r_2, r_3\}$, are arranged in an equilateral triangle as Figure \ref{impos_eqtri},
and robot $r_1$ (resp. $r_2$ and $r_3$) observes $r_2$ (resp. $r_3$ and $r_1$).
All robots do not agree on any geometrical agreement 
(\eg direction, orientation, chirality, or unit distance),
thus we can assume that every robot $r_i$ considers the direction to the center of the triangle ($p_c$) from itself 
(\ie $\overrightarrow{r_ip_c}$) as the positive direction of $X$-axis 
in its local coordinate system. 
Moreover, we also assume that all robots have the same chirality (\eg clockwise) and the same unit distance.
This means that all robots have the exactly same result of \LL~operation.

Let $\mathcal{A}$ be an algorithm for gathering in the distance-based (3,1)-defected model.
In the above configuration, 
all robots have the same observation results, 
thus they execute the same behaviors according to $\mathcal{A}$
(\ie all robots move to the same $x$ and $y$ coordinates in their local coordinate systems).
This causes another configuration forming a different equilateral triangle, which shows 
by repeating the argument that the robots cannot gather at the same point forever.
The only way to prevent the robots from forming another equilateral triangle is to move to point $p_c$, 
\ie each robot moves to the point located at $|\overline{r_ir_j}| / \sqrt{3}$ 
distance in the $30^\circ$ clockwise direction of the observed robot $r_j$.
However, if all robots agree on the opposite direction of chirality 
(counter-clockwise in this case),
they move to the outside of triangle $\triangle{r_1r_2r_3}$
(\ie robot $r_1$ moves to point $p_c'$ instead of $p_c$).
As a result, the robots form another equilateral triangle.
\end{proof}

\subsection{Impossibility in the relaxed adversarial ($N$,$N-2$)-defected model}
\label{sec:relaxed}

The ($N$, $k$)-defected model assumes that $k$ robots observed by robot $r$ are chosen from the robots that are located at points other than $r$'s current position and that $r$ can detect whether it is single or accompanied.
Natural relaxation of the model is to choose the $k$ robots other than $r$ (\ie robots at $r$'s current position can be chosen) and assume the weak multiplicity detection for the $k$ robots and $r$ itself. 
We call the model with the relaxation the \emph{relaxed adversarial} ($N$,$k$)-defected model.
Notice that the key property of the ($N$, $N-2$)-defected model such that any accompanied robot can observe all the robots does not hold in the relaxed model.

The following theorem shows that the gathering is impossible (from some configuration) in the relaxed adversarial ($N$,$N-2$)-defected model.

\begin{theorem}
There is no (deterministic) algorithm to solve the gathering problem in the relaxed adversarial ($N$,$N-2$)-defected model.
\end{theorem}

\begin{proof}
Let $\mathcal{A}$ be a gathering algorithm in the relaxed adversarial ($N$,$N-2$)-defected model.
We consider only initial configurations where all robots are located at two points $p_1$ and $p_2$.

First, consider the initial configuration where $N-1$ robots are located at $p_1$ and one robot, say $r_1$, is located at $p_2$.  When the robots at $p_1$ do not observe $r_1$, they misunderstand that the gathering is already achieved and terminate.  To achieve the gathering, $r_1$ has to move to $p_1$.  
This implies that $\mathcal{A}$ has the following action ({\bf Action 1}):
when a single robot $r$ observes only one occupied point other than $r$'s current point and recognizes that the point is occupied by multiple robots, $r$ has to move to the point. 

Notice that {\bf Action 1} is sufficient to show that $\mathcal{A}$ cannot solve the gathering in the relaxed adversarial (4,2)-defected model.  Consider the initial configuration where two robots exist at both of $p_1$ and $p_2$ (four robots in total).  When the robots at $p_1$ (resp. $p_2$) observe only the two robots at $p_2$ (resp. $p_1$), the robots at $p_1$ (resp. $p_2$) move to $p_2$ (resp. $p_1$) by {\bf Action 1}.
At the resultant configuration, two robots exist at both of $p_1$ and $p_2$, which shows by repeating the argument that algorithm $\mathcal{A}$ cannot solve the gathering problem.

Second, consider the initial configuration where $N\ge 5$ and all robots recognize that both $p_1$ and $p_2$ are occupied by multiple robots, which can occur when a point is occupied by three or more robots and the other is occupied by two or more robots.
When all the robots at the same point observe the same set of robots (but still they recognize that both the points are occupied by multiple robots), the robots at the same point execute the same action (\ie move to the same point).  
Since algorithm $\mathcal{A}$ solves the gathering problem, all robots eventually have to move to the same point (precisely the midpoint of the two points occupied by robots) to achieve the gathering.
This implies that $\mathcal{A}$ has the following action ({\bf Action 2}):
when an accompanied robot $r$ observes only one occupied point other than $r$'s current point and recognizes that the point is occupied by multiple robots, $r$ has to move to the midpoint of the two points. 

Finally, consider the initial configuration of $N\ (\ge 5)$ robots where two robots exist at $p_1$ and $N-2$ robots exist at $p_2$.
When each robot $r_1$ at $p_1$ observes only $N-2$ robots at $p_2$ (and recognizes itself as a single robot), $r_1$ moves to $p_2$ by {\bf Action 1}.
On the other hand, when each robot $r_2$ at $p_2$ observes the two robots at $p_1$ and $N-4$ robots (other than $r_2$) at $p_2$, $r_2$ moves to the midpoint of $p_1$ and $p_2$ by {\bf Action 2}.
At the resultant configuration, two robots exist at $p_2$ and $N-2$ robots exist at the midpoint of $p_1$ and $p_2$.
By repeating the argument, we can show that algorithm $\mathcal{A}$ cannot solve the gathering problem although all the robots converges at the same point (\ie the distance between the two groups of robots becomes smaller and smaller but does not become zero).

Consequently, there is no gathering algorithm in the relaxed adversarial ($N$,$N-2$)-defected model.
\end{proof}

\section{Conclusion and Open Problems}
\label{sec:conclusion}
In this paper, we introduced a new computational model, the ($N$, $k$)-\emph{defected model}, 
where each robot cannot necessarily observe all other robots:
\ie each robot observes at most $k$ other robots not located at its current position (where $k < N-1$).
We addressed the gathering problem, which is one of the basic problem in autonomous mobile robot systems, 
in the ($N$, $N-2$)-defected model.
We proposed two gathering algorithms:
(1) an algorithm in the adversarial ($N$,$N-2$)-defected model that achieves the gathering in three rounds,
and (2) an algorithm in the distance-based (4,2)-defected model that achieves the gathering in four rounds.
Moreover, we showed that there is no (deterministic) algorithm in either the adversarial or distance-based (3,1)-defected model.
In the proposed model, 
we assume that each robot $r$ observes $k$ other robots among the robots located at the different points 
than the point occupied by $r$ itself.
The relaxation of this assumption, where $k$ robots are chosen among all other robots other than $r$,
can be considered, however, we proved that the gathering is unsolvable in this relaxed model.

The remaining problem we are most interested in is to clarify the solvability of the gathering problem 
in the adversarial (4,2)-defected model.
Remind that the basic strategy of the proposed algorithm in the distance-based (4,2)-defected model is 
to determine one unique point from the triangle formed by the observed set of points.
We call the algorithm using this strategy the \emph{set-based algorithm},
where each robot determines the destination referring to only the set of observed points:
for example, when a robot observes an isosceles triangle, 
it always moves to the midpoint of the base, whether it is adjacent to the base or not,
\ie we do not use the information of the (relative) position of the observing robot.
We can prove that \emph{there is no (deterministic) set-based algorithm to solve the gathering problem in adversarial (4,2)-defected model}.
This means that if there is a gathering algorithm in the adversarial (4,2)-defected model,
each robot uses its relative position in the set of observed points,
\eg when a robot observes an isosceles triangle, the destination point changes 
whether the robot is at a point incident to the base of the triangle or not.

An important future work is to find the minimum $k$ that allows a solution for the gathering problem in the adversarial or distance-based ($N$, $k$)-defected model.
In this paper, we considered only the gathering problem, 
thus to challenge other problems under the ($N$, $k$)-defected model is another future work.



\bibliography{nkgathering-bib-short}

\begin{thebibliography}{10}

\bibitem{gathermulfault}
Zohir Bouzid, Shantanu Das, and S{\'{e}}bastien Tixeuil.
\newblock Gathering of mobile robots tolerating multiple crash faults.
\newblock In {\em {IEEE} 33rd International Conference on Distributed Computing
  Systems, {ICDCS}}, pages 337--346. {IEEE} Computer Society, 2013.
\newblock \href {https://doi.org/10.1109/ICDCS.2013.27}
  {\path{doi:10.1109/ICDCS.2013.27}}.

\bibitem{gathernulimited}
Avik Chatterjee, Sruti~Gan Chaudhuri, and Krishnendu Mukhopadhyaya.
\newblock Gathering asynchronous swarm robots under nonuniform limited
  visibility.
\newblock In {\em Distributed Computing and Internet Technology - 11th
  International Conference, {ICDCIT}}, volume 8956 of {\em Lecture Notes in
  Computer Science}, pages 174--180. Springer, 2015.
\newblock \href {https://doi.org/10.1007/978-3-319-14977-6\_11}
  {\path{doi:10.1007/978-3-319-14977-6\_11}}.

\bibitem{gatheranyconf}
Mark Cieliebak, Paola Flocchini, Giuseppe Prencipe, and Nicola Santoro.
\newblock {Solving the Robots Gathering Problem}.
\newblock In {\em Proceedings of the 30th International Colloquium on Automata,
  Languages and Programming, {ICALP}}, pages 1181--1196. Springer, 2003.
\newblock \href {https://doi.org/10.1007/3-540-45061-0_90}
  {\path{doi:10.1007/3-540-45061-0_90}}.

\bibitem{gathergrid}
Gianlorenzo D'Angelo, Gabriele~Di Stefano, Ralf Klasing, and Alfredo Navarra.
\newblock Gathering of robots on anonymous grids and trees without multiplicity
  detection.
\newblock {\em Theor. Comput. Sci.}, 610:158--168, 2016.
\newblock \href {https://doi.org/10.1016/j.tcs.2014.06.045}
  {\path{doi:10.1016/j.tcs.2014.06.045}}.

\bibitem{gatherfault}
Xavier D{\'{e}}fago, Maria Gradinariu, St{\'{e}}phane Messika, and
  Philippe~Raipin Parv{\'{e}}dy.
\newblock {Fault-Tolerant and Self-stabilizing Mobile Robots Gathering}.
\newblock In {\em Proceedings of the 20th International Symposium on
  Distributed Computing, {DISC}}, pages 46--60. Springer, 2006.
\newblock \href {https://doi.org/10.1007/11864219_4}
  {\path{doi:10.1007/11864219_4}}.

\bibitem{gatherlimited}
Paola Flocchini, Giuseppe Prencipe, Nicola Santoro, and Peter Widmayer.
\newblock {Gathering of Asynchronous Oblivious Robots with Limited Visibility}.
\newblock In {\em Proceedings of the 18th Annual Symposium on Theoretical
  Aspects of Computer Science, {STACS}}, pages 247--258. Springer, 2001.
\newblock \href {https://doi.org/10.1007/3-540-44693-1_22}
  {\path{doi:10.1007/3-540-44693-1_22}}.

\bibitem{gatherinac3}
Nobuhiro Inuzuka, Yuichi Tomida, Taisuke Izumi, Yoshiaki Katayama, and Koichi
  Wada.
\newblock {Gathering Problem of Two Asynchronous Mobile Robots with
  Semi-dynamic Compasses}.
\newblock In {\em Proceedings of the 15th International Colloquium on
  Structural Information and Communication Complexity, {SIROCCO}}, pages 5--19.
  Springer, 2008.
\newblock \href {https://doi.org/10.1007/978-3-540-69355-0_3}
  {\path{doi:10.1007/978-3-540-69355-0_3}}.

\bibitem{gatherinac2}
Taisuke Izumi, Yoshiaki Katayama, Nobuhiro Inuzuka, and Koichi Wada.
\newblock {Gathering Autonomous Mobile Robots with Dynamic Compasses: An
  Optimal Result}.
\newblock In {\em Proceedings of the 21st International Symposium on
  Distributed Computing, {DISC}}, pages 298--312. Springer, 2007.
\newblock \href {https://doi.org/10.1007/978-3-540-75142-7_24}
  {\path{doi:10.1007/978-3-540-75142-7_24}}.

\bibitem{gatherinac4}
Yoshiaki Katayama, Yuichi Tomida, Hiroyuki Imazu, Nobuhiro Inuzuka, and Koichi
  Wada.
\newblock {Dynamic Compass Models and Gathering Algorithms for Autonomous
  Mobile Robots}.
\newblock In {\em Proceedings of the 14th International Colloquium on
  Structural Information and Communication Complexity, {SIROCCO}}, pages
  274--288. Springer, 2007.
\newblock \href {https://doi.org/10.1007/978-3-540-72951-8_22}
  {\path{doi:10.1007/978-3-540-72951-8_22}}.

\bibitem{convlimited}
David~G. Kirkpatrick, Irina Kostitsyna, Alfredo Navarra, Giuseppe Prencipe, and
  Nicola Santoro.
\newblock Separating bounded and unbounded asynchrony for autonomous robots:
  Point convergence with limited visibility.
\newblock In {\em {PODC} '21: {ACM} Symposium on Principles of Distributed
  Computing}, pages 9--19. {ACM}, 2021.
\newblock \href {https://doi.org/10.1145/3465084.3467910}
  {\path{doi:10.1145/3465084.3467910}}.

\bibitem{gatherring}
Ralf Klasing, Euripides Markou, and Andrzej Pelc.
\newblock Gathering asynchronous oblivious mobile robots in a ring.
\newblock {\em Theor. Comput. Sci.}, 390(1):27--39, 2008.
\newblock \href {https://doi.org/10.1016/j.tcs.2007.09.032}
  {\path{doi:10.1016/j.tcs.2007.09.032}}.

\bibitem{circlelimited}
Giuseppe Antonio~Di Luna, Ryuhei Uehara, Giovanni Viglietta, and Yukiko
  Yamauchi.
\newblock Gathering on a circle with limited visibility by anonymous oblivious
  robots.
\newblock In {\em 34th International Symposium on Distributed Computing,
  {DISC}}, volume 179 of {\em LIPIcs}, pages 12:1--12:17. Schloss Dagstuhl -
  Leibniz-Zentrum f{\"{u}}r Informatik, 2020.
\newblock \href {https://doi.org/10.4230/LIPIcs.DISC.2020.12}
  {\path{doi:10.4230/LIPIcs.DISC.2020.12}}.

\bibitem{gatherweakmulfault}
Debasish Pattanayak, Kaushik Mondal, H.~Ramesh, and Partha~Sarathi Mandal.
\newblock Fault-tolerant gathering of mobile robots with weak multiplicity
  detection.
\newblock In {\em Proceedings of the 18th International Conference on
  Distributed Computing and Networking}, page~7. {ACM}, 2017.
\newblock URL: \url{http://dl.acm.org/citation.cfm?id=3007786}.

\bibitem{gatherfea}
Giuseppe Prencipe.
\newblock On the feasibility of gathering by autonomous mobile robots.
\newblock In {\em Proceedings of the 12th International Colloquium on
  Structural Information and Communication Complexity, {SIROCCO}}, pages
  246--261. Springer, 2005.
\newblock \href {https://doi.org/10.1007/11429647_20}
  {\path{doi:10.1007/11429647_20}}.

\bibitem{gatherimpos}
Giuseppe Prencipe.
\newblock Impossibility of gathering by a set of autonomous mobile robots.
\newblock {\em Theor. Comput. Sci.}, 384(2-3):222--231, 2007.
\newblock \href {https://doi.org/10.1016/j.tcs.2007.04.023}
  {\path{doi:10.1016/j.tcs.2007.04.023}}.

\bibitem{perspec13}
Giuseppe Prencipe.
\newblock {Autonomous Mobile Robots: {A} Distributed Computing Perspective}.
\newblock In {\em Proceedings of the 9th International Symposium on Algorithms
  and Experiments for Sensor Systems, Wireless Networks and Distributed
  Robotics, {ALGOSENSORS}}, pages 6--21. Springer, 2013.
\newblock \href {https://doi.org/10.1007/978-3-642-45346-5_2}
  {\path{doi:10.1007/978-3-642-45346-5_2}}.

\bibitem{distalg06}
Giuseppe Prencipe and Nicola Santoro.
\newblock {Distributed Algorithms for Autonomous Mobile Robots}.
\newblock In {\em Proceedings of Fourth {IFIP} International Conference on
  Theoretical Computer Science, {TCS}}, pages 47--62. Springer, 2006.
\newblock \href {https://doi.org/10.1007/978-0-387-34735-6_8}
  {\path{doi:10.1007/978-0-387-34735-6_8}}.

\bibitem{gatherinac1}
Samia Souissi, Xavier D{\'{e}}fago, and Masafumi Yamashita.
\newblock {Gathering Asynchronous Mobile Robots with Inaccurate Compasses}.
\newblock In {\em Proceedings of the 10th International Conference on
  Principles of Distributed Systems, {OPODIS}}, pages 333--349. Springer, 2006.
\newblock \href {https://doi.org/10.1007/11945529_24}
  {\path{doi:10.1007/11945529_24}}.

\bibitem{SY99}
Ichiro Suzuki and Masafumi Yamashita.
\newblock {Distributed Anonymous Mobile Robots: Formation of Geometric
  Patterns}.
\newblock {\em {SIAM} J. Comput.}, 28(4):1347--1363, 1999.
\newblock \href {https://doi.org/10.1137/S009753979628292X}
  {\path{doi:10.1137/S009753979628292X}}.

\end{thebibliography}

\end{document}